\documentclass[prodmode,acmtsas]{acmsmall} 
\usepackage{multirow}
\usepackage{graphicx}
\usepackage{clrscode3e}
\usepackage{algorithmic}
\usepackage{cite}
\usepackage{amssymb}
\usepackage{amsmath}
\usepackage{ctable}

\usepackage{wrapfig,lipsum}
\usepackage{epsfig}
\newcommand{\argmax}{\arg\!\max}

\newcommand{\alg}{{\sc TopCom}}
\usepackage{algorithm}

\usepackage{subcaption}
\begin{document}
\markboth{V. S. Dave et al.}{TopCom}

\title{TopCom: Index for Shortest Distance Query in Directed Graph}
\author{VACHIK S. DAVE
\affil{Indiana University Purdue University, Indianapolis}
MOHAMMAD AL HASAN
\affil{Indiana University Purdue University, Indianapolis}
}


\begin{abstract}
Finding shortest distance between two vertices in a graph is an important
problem due to its numerous applications in diverse domains, including
geo-spatial databases, social network analysis, and information retrieval.
Classical algorithms (such as, Dijkstra) solve this problem in polynomial time,
but these algorithms cannot provide real-time response for a large number of
bursty queries on a large graph. So, indexing based solutions that pre-process
the graph for efficiently answering (exactly or approximately) a large number of
distance queries in real-time is becoming increasingly popular. Existing
solutions have varying performance in terms of index size, index building time,
query time, and accuracy.  In this work, we propose \alg, a novel
indexing-based solution for exactly answering distance queries. Our experiments
with two of the existing state-of-the-art methods (IS-Label and TreeMap) show
the superiority of \alg\ over these two methods considering scalability and
query time. Besides, indexing of \alg\ exploits the DAG (directed acyclic
graph) structure in the graph, which makes it significantly faster than the
existing methods if the SCCs (strongly connected component) of the input graph
are relatively small.
\end{abstract}

%
%


%
%

\terms{Graph Algorithms, Performance}



\begin{bottomstuff}
Author's addresses: V. S. Dave, Computer \& Information Science Department,
Indian University Purdue University, Indianapolis; M. Al Hasan,
Computer \& Information Science Department,
Indian University Purdue University, Indianapolis.
\end{bottomstuff}

\maketitle

\section{Introduction}

\label{sec:intro}

Finding shortest distance between two nodes in a graph (\textit{distance
query}) is one of the most useful operations in graph analysis. Besides the
application that stands for its literal meaning, i.e.\ finding the shortest
distance between two places in a road network, this operation is useful in many
other applications in social and information networks. For instance, in social
networks, the shortest path distance is used in the calculation of different
centrality metrics, including closeness centrality and betweenness
centrality~\cite{RCC:Preparata:2008,IAN:Erdem:2013}.
It is also used as a criterion for finding highly influential
nodes~\cite{MSI:Kempe:2003}, and for detecting communities in a
network~\cite{GFL:Backstrom:2006}. Scientists have also used shortest path
distance to generate features for predicting future links in a
network~\cite{ASL:Hasan:2011}. In information networks, shortest path distance
is used for keyword search~\cite{KSG:Kargar:2011}, and also for relevance
ranking~\cite{SWC:Ukkonen:2008}.

Due to the importance of the shortest path distance problem, researchers have been
studying this problem from the ancient time, and several classical algorithms
(Dijkstra, Bellman-Ford, Floyd-Warshall) exist for this problem, which run in
polynomial time over the number of vertices and the number of edges of the
network. However, as real-life graphs grow in the order of thousands or
millions of vertices, classical algorithms deem inefficient for providing
real-time answers for a large number of distance queries on such graphs. For
example, for a graph of a few thousand vertices, a contemporary desktop
computer takes an order of seconds to answer a single query, so thousands of
queries take tens of minutes, which is not acceptable for many real-time
applications. So, there is a growing interest for the discovery of more
efficient methods for solving this task.

Various approaches are considered for obtaining an efficient distance query method
for large graphs. One of them is to exploit topological properties of real-life
networks that adhere to some specific characteristics. For instance, many
researchers exploit the spatial and planar properties of road
networks~\cite{KSP:Tao:2011,HLA:Abraham:2011,FDP:Yan:2013} to obtain efficient
solutions for distance queries in road networks.  However, for a general
network from any other domain, such methods perform
poorly~\cite{HHL:Abraham:2012}. The second approach is to perform
pre-processing on the host graph and build an index data structure which can be
used at runtime to answer the distance query between an arbitrary pair of nodes
more efficiently. Several indexing ideas are used, but two are the most common,
landmark-based
indexing~\cite{FFD:Tretyakov:2011,ASD:Qiao:2014,FSP:Potamias:2009,FES:Akiba:2013}
and 2-hop-cover indexing~\cite{2Hop:Cohen:2002}. Methods adopting the former
idea identify a set of landmark nodes and pre-compute all-single source
shortest paths from these landmark nodes. During query time, distances between
a pair of arbitrary nodes are answered from their distances to their respective
closest landmark nodes. Most of these methods deliver an approximation of
shortest path distance except a method presented in ~\cite{FES:Akiba:2013}.
Methods adopting the two-hop cover indexing generally find the exact solution
for a distance query~\cite{AHC:Jin:2012,HDL:Jiang:2014,ISL:Fu:2013}. These
methods store a collection of hops (paths starting from that node), such that
the shortest path between a pair of arbitrary vertices can be obtained from the
intersection of the hops of those vertices.

A related work to the shortest path problem is the reachability problem. Given
a directed graph $G(V, E)$, and a pair of vertices $u$ and $v$, the
reachability problem answers whether a path exists from $u$ to $v$. This
problem can be solved in $O(|V|+|E|)$ time using graph traversal, where $V$ is
the set of vertices and $E$ is the set of edges. However, using a reachability
index, a better runtime can be obtained in practice. All the existing
solutions~\cite{GRAIL:Hilmi:2012,RQL:Zhu:2014} of the reachability problem
solve it for a directed acyclic graph (DAG). This is due to the fact that any
directed graph can be converted to a DAG such that a DAG node is a strongly
connected component (SCC) of the original graph; since any nodes in an SCC is
reachable to each other, the reachability solution in the DAG easily answers
a reachability query in the original graph. The indexing idea that we
propose in this work also exploits the SCC, but unlike existing works we solve
the distance query problem instead of reachability.

In this work, we propose \alg~\footnote{\alg\ stands for {\bf Top}ological {\bf
Com}pression which is the fundamental operation that is used to create the
index data structure of this method.}, an indexing based method for obtaining
exact solution of a distance query in an arbitrary directed graph. In principle,
\alg\ uses a 2-hop-cover solution, but its indexing is different from other
existing indexing methods. Specifically, the basic indexing scheme of \alg\ is
designed for a DAG and it is inspired from the indexing solution of the
reachability queries proposed in~\cite{TFL:Cheng:2013}. Due to its design,
\alg\ exhibits a very attractive performance for a DAG or general graph
in which SCCs are relatively small. However, we also extend the basic indexing
scheme so that it also solves the distance query for an arbitrary directed graph. We
show experiment results that validate \alg's superior performance over IS-Label
\cite{ISL:Fu:2013} and TreeMap~\cite{AED:Xiang:2014} which are two of the fastest known
indexing based shortest path methods in the recent years. Following other recent works, 
we also compare our method with bi-directional Dijkstra, which is a well-accepted baseline 
method for distance query solutions in a directed graph. Note that, this journal
article is an extended version of a published conference article~\cite{TIS:dave:2015};
the conference article works for DAG only, but this work solves distance query
indexing for arbitrary directed graphs.

\section{Related works}
\label{sec:related-works}
Shortest distance on a graph has many interesting recent and earlier works.  In
the section we discuss the most important works among these under two 
categories: (1) Online shortest distance calculation, (2) Offline
(Index based) shortest distance calculation.

\subsection{Online shortest distance calculation:}
\label{sec:online-methods}

For unweighted graph, the simplest online method to find shortest distance is
Breadth First Search (BFS) with time complexity $O(|V| + |E|)$, where $V$ is number
of vertices and $E$ is number of edges of the graph. For weighted graph, most
well-known single source shortest distance algorithm is Dijkstra's algorithm,
which computes shortest distance for weighted graph with positive weights.
Using a binary heap based priority queue, the complexity of Dijkstra's
algorithm is $O(|E| \lg |V|)$ and the same using a Fibonacci heap is $O(|E|+|V| \lg
|V|)$.  Another well known algorithm for single source shortest path is the
Bellman-Ford algorithm~\cite{ORP:bellman:1958,CSRL:2001} with time complexity
$O(|V|\cdot|E|)$, which is generally slow for large graphs with millions of nodes and
edges. 

There are methods proposed by different researchers to improve the above
classical shortest distance methods~\cite{CHG:Bauer:2010,SUT:Wanger:2007}.
Although, they do not improve the worst case complexity of the shortest path
algorithm, they do exhibit good average-case behavior. The most popular among
these methods is Bidirectional Dijkstra~\cite{IBH:Sint:1977}, which is particularly
applicable when the objective is to obtain the shortest distance between a pair
of vertices. The computational complexity of bidirectional search can be
denoted as $O(b^{d/2})$, where $b$ is the branching factor and $d$ is the
distance from start node to target node. Real life networks have small value of
$d$ (typically smaller than 10)---a fact that makes this algorithm an attractive
choice for many applications. In this work, bidirectional Dijkstra is one of the
methods with which we compare our proposed solution.

\subsection{Offline (Index based) shortest distance  calculation:}
\label{sec:offline-methods} 

For large graphs, online methods are slower than an indexing based method, so
most of the recent research efforts are concentrated towards indexing based
methods. The literature  for shortest distance indexing is quite vast,
so, we review few of the works that have published in the recent years. For a
detailed review, we refer the readers to read ~\cite{EAD:Zwick:2001,SPQ:Sommer:2014}. 


Many of the existing works for shortest distance computation is specifically
designed for the {\bf road 
networks}~\cite{FDP:Yan:2013,GIR:Rice:2010,KSP:Tao:2011,SPD:Zhu2013,HLA:Abraham:2011,CHF:Geisberger:2008,HHH:Sanders:2005,EPC:Jung:2002}. 
Such networks show hierarchical structures with the presence of
junctions, hubs, and highways; the shortest distance computation methods for
these networks exploit the hierarchical structure for compressing distance
matrix or for building distance indices~\cite{GIR:Rice:2010,CHF:Geisberger:2008}. 
For example, Sanders et al.~\cite{HHH:Sanders:2005} use
highway hierarchy and design an exact shortest distance computation method that
runs faster than a method that does not use the hierarchy structure. Zhu
et al.~\cite{SPD:Zhu2013} design a hierarchy based indexing and prove that the
results on real-life graphs is close to the theoretical complexity of the
proposed method. Jung et al.~\cite{EPC:Jung:2002} design an efficient shortest
path computation method for hierarchically structured topographical road maps.
Abraham et al.~\cite{HLA:Abraham:2011}  have proposed an efficient hub-based
labeling (HL) method to answer shortest path distance query on road networks.
Tao et al.~\cite{KSP:Tao:2011} explore the spatial property and find k-skip
graph which can answer k-skip shortest path i.e. path created from the original
shortest path by skipping at-most k consecutive nodes. Recently Yan
et al.~\cite{FDP:Yan:2013} propose a method to find the distance preserving
sub-graphs to answer a shortest distance query more efficiently. However, most
of the indexing schemes for the road networks are based on some specific
property of the road networks and they are ineffective for general graphs that
do not satisfy those properties of road networks~\cite{HHL:Abraham:2012}.

Finding exact shortest distance in a large graph is a costly task, hence few
researchers have proposed methods for computing \textbf{estimated shortest distance}
~\cite{FFD:Tretyakov:2011,FAE:Gubichev:2010,ASD:Qiao:2014,FSP:Potamias:2009}. 
The most common among the estimated shortest distance based
methods is the landmark based method, which selects a set of landmark nodes
based on some criteria and finds shortest paths that must go through those
landmark nodes.  The main task here is to decide the set of vertices that are
optimal choice as landmarks. However, it has been shown that this optimization
problem is $\cal{NP}$-Hard~\cite{FSP:Potamias:2009}, so researchers adopt
various heuristics based approaches for choosing those landmarks.  Potamias et
al.~\cite{FSP:Potamias:2009} compare centrality and degree based approaches for
selecting landmarks. Gubichev et al.~\cite{FAE:Gubichev:2010} propose a sketch
based indexing method for estimating answer of a shortest distance query.
Treyakov et al.~\cite{FFD:Tretyakov:2011} propose a landmark based fully
dynamic approximation method using shortest path tree and also obtain an
improved set of vertices as landmark; they show that their method has less
approximation error than other landmark based approaches. Qiao et
al.~\cite{ASD:Qiao:2014} propose a query based local landmark method which
selects landmark nodes that are local to the query in the sense that the
obtained shortest path is the closest to the real shortest path as much as
possible; this method also improves the estimation accuracy. Our proposed
indexing method, \alg, is not comparable to these methods,
because unlike these methods, our method provides exact shortest path distance.

There are some other works for finding shortest distance in large graphs which
are proposed very recently; examples include~\cite{ISL:Fu:2013,ESS:Zhu:2013,AHC:Jin:2012,OLE:Cheng:2009,HHL:Abraham:2012,AED:Xiang:2014,FES:Akiba:2013,RAS:Gao:2011,TES:Wei:2010,EPD:Cheng:2012}. Many of these have
unique ideas, so it is difficult to categorize them under a generic shortest
path method. For example, Gao et al.~\cite{RAS:Gao:2011} use a relational
approach and propose an index called SegTable which stores local segments of a
shortest distance. Zhu at al.~\cite{ESS:Zhu:2013} propose a method to answer
single source shortest distance query for a huge graph on disk. Akiba et 
al.~\cite{FES:Akiba:2013}~\footnote{This work cannot be compared with \alg, 
because the authors were unable to provide code that can answer shortest distance 
query in directed graphs.}
propose a unique pruning method based on degree of a vertex, which can efficiently reduce
the search space of BFS.  Highway centric label (HCL) \cite{AHC:Jin:2012} is
one of the fastest recent methods that is proposed for a shortest distance
query on both directed and undirected graphs.  In a follow-up work, Xiang
proposes TreeMap ~\cite{AED:Xiang:2014}, a tree decomposition based approach
for solving distance query exactly; the author compares TreeMap's solution with
those of HCL to show that the former has better performance.  Another recent
method is called IS-Label which is proposed by  Fu et al.~\cite{ISL:Fu:2013}.
They have also shown that that IS-Label has superior performance than HCL.  In
this work, we compare \alg\ with both IS-Label and TreeMap, which are among the
best of the existing index based methods.

Our proposed method exploits conversion of a directed graph into a {\em
directed acyclic graph (DAG)} by collapsing {\em strongly connected components}
(SCCs) into a vertex. It is a widely used approach for solving {\bf
reachability query} task in a directed graph~\cite{TFL:Cheng:2013,3HOPP:Jin:2009,GRAIL:Hilmi:2012,EAR:Jin:2008,SSR:Jin:2012,OLE:Cheng:2009}.
Any two nodes in an SCC are reachable from each other, hence for a directed
graph the reachability query between two nodes can be answered through the
reachability answer between their corresponding DAG nodes. However, note that,
a reachability query is easier than a distance query, because the latter
provides distance value as answer, which is relatively harder as the graph can
be weighted. Specifically, for online (non-indexing) solution, reachability has
a complexity of $O(|V|+|E|)$, and shortest distance query for weighted graph
has a higher complexity, $O(|E| \lg |V|)$.  Authors in~\cite{TFL:Cheng:2013} uses
topological folding of DAG for answering reachability query. Our proposed
method \alg\ also uses topological folding property of DAG to compress the DAG
level-wise, but unlike the above work, our work answers distance queries
on weighted graph.  Cheng et al.~\cite{OLE:Cheng:2009} is another existing work
which also proposes a DAG based approach for answering distance queries by
finding distance aware 2-HOP cover.

Note that, one of the known limitations of indexing based methods is that they
require more memory, but this is not a concern for \alg\ with today's commodity machine
having main memory in multiples of 2 GB.

\begin{figure*} [t]
\centering
\begin{subfigure}[h]{0.7\textwidth}
	\includegraphics[width=\textwidth]{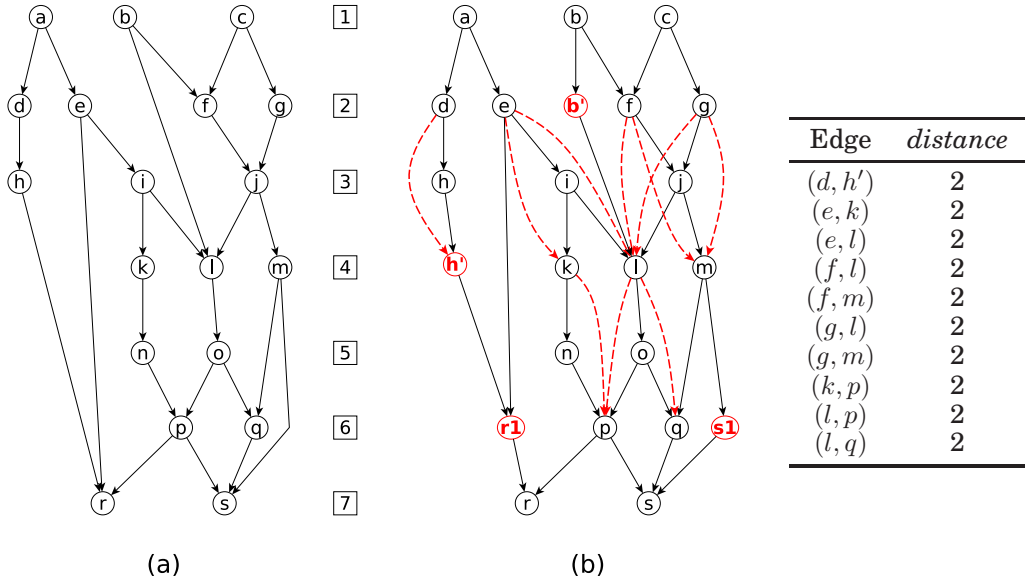}
\end{subfigure}
~
\begin{subfigure}[h]{0.28\textwidth}
\centering

  \begin{tabular}{cc} 
\toprule
  Edge & \textit{distance} \\ 
\midrule 
  $(d,h')$ & 2 \\
  $(e,k)$ & 2 \\
  $(e,l)$ & 2 \\
  $(f,l)$ & 2 \\
  $(f,m)$ & 2 \\
  $(g,l)$ & 2 \\
  $(g,m)$ & 2 \\
  $(k,p)$ & 2 \\
  $(l,p)$ & 2 \\
   $(l,q)$ & 2 \\
\bottomrule
  
  \end{tabular}
\end{subfigure}
\caption{Pre-processing of DAG before Compression: (a) Original DAG $G$ and (b) Modified DAG $G_{m}$. 
The dummy edges data structure ($DummyEdges$) associated with this modified DAG is shown to the right.}
\label{fig:dag-eg}
\end{figure*}

\section{Method}
\label{sec:method}
In this section, we discuss the shortest distance indexing of \alg\ for a DAG.
In subsequent section, we will show how this can be adapted for a general
directed graph.

\subsection{Topological compression} 
\label{sec:tc}

The main idea of \alg\ is based on topological compression of DAG, which is
performed during the index building step. During the compression, additional
distance information is preserved in a data structure which \alg\ uses for
answering a distance query efficiently. For the sake of simplicity, in
subsequent discussion we assume that the given graph is unweighted for which
the weight of each edge is 1 and the distance between two vertices is the
minimum hop count between them.  We will discuss the necessary adaptations that
are needed for a weighted graph at the end of this section.

\noindent {\bf Topological Level:} Given a DAG $G$, we use $V(G)$ and $E(G)$ to
represent set of vertices and edges of $G$, respectively. The topological
level of any vertex $v \in V(G)$, defined as $topo(v)$, is $1$ if $v$ has no
incoming edge, otherwise it is at least one higher than the topological level
of any of $v$'s parents. Mathematically, \[ topo(v)= \begin{cases}
\displaystyle \max_{(u, v)\in E(G)}~~ {topo(u)}+1,	& \text{if $v$ has
incoming edges}\\ 1,         	& \text{otherwise} \end{cases} \]

For a vertex $v$, if $topo(v)$ is even, we call $v$ an even-topology vertex,
otherwise $v$ is an odd-topology vertex.  An edge, $e = (u, v) \in E(G)$, is a
single-level edge if $topo(v)-topo(u)=1$, otherwise it is a multi-level edge.
For a DAG $G$, its topological level is the largest value of $topo(v)$ over the
vertices in $G$, i.e.: \[topo(G) = \max_{v \in V(G)}~~ topo(v)\]

\noindent {\bf Example:} Consider the DAG $G$ in Figure~\ref{fig:dag-eg}(a).
Topological level of vertices, $a, b$, and $c$ is 1, as the vertices have no incoming
edge. The topological level of vertex $l$ is 4, as one of the predecessor node
of $l$ is $i$, which has a topological level value of 3.  $Topo(G)$ is equal to
7, because 7 is the largest topological level value for one of the vertices in
$G$.~$\qed$

Topological compression of a DAG is performed iteratively, such that the
compressed output of one iteration is the input of subsequent iteration.  For an
input DAG $G$, one iteration of topological compression removes all
odd-topology vertices from $G$ along with the edges that are incident to the
removed vertices.  All single-level edges are thus removed, as one of the
adjacent vertices of these edges is an odd-topology vertex.  A multi-level edge
is also removed if at least one of the endpoints of the edge is an
odd-topology vertex. As a result of this compression, the topological level of
$G$ reduces by half. For the purpose of shortest distance index building,
starting from $G=G^0$, we apply this compression process iteratively to
generate a sequence of DAGs $G^1, G^2, \cdots, G^t$ such that the topological
level number of each successive DAG is half of that of the previous DAG, and
the topological level number of the final DAG in this sequence is 1; i.e.,
$topo(G^{i+1})= \lfloor topo (G^{i})/2 \rfloor$, and $topo(G^t)=1$, where $t =
\lfloor \log_{2} {topo(G)} \rfloor$.\\

\noindent {\bf Example:} Consider the same DAG $G=G^0$ in
Figure~\ref{fig:dag-eg}(a).  Its topological compression in the first
iteration, $G^1$ is shown in Figure~\ref{fig:comgraph}(a), and in the second
iteration, $G^2$ is shown in Figure~\ref{fig:comgraph}(c). $G^2$ is the last
compression state of $G^0$, as topological level of $G^2$ is 1. Note that,
in $G^1$, all odd-topology vertices of $G^0$, such as, $a, b, c, h, i,j$, etc.
are removed. All single-level edges of $G^0$, such as, $(e,i), (k,n), (p,s)$, etc.
are removed. Multi-level edges, such as, $(b,l)$ and $(m,s)$ are also removed.
However, there are newly added vertices in $G^1$, such as $b', h', r1$, and $s1$,
along with newly added edges, such as, $(b',l)$ and $(e,r1)$. More discussion
about these additional vertices and edges are given in the following paragraphs.
~$\qed$

Each iteration of topological compression of a DAG causes loss of information
regarding the connectivity among the vertices; for correctly answering distance
queries \alg\ needs to preserve the connectivity information as the input DAG
is being compressed. The preservation process gives rise to additional vertices
and edges in $G^1$, which we have seen in the above example. The connectivity
preservation process is discussed in detail below.

The most common information loss is caused by the removal of single-level
edges. However, such edges are also easily recoverable from the lastly
compressed graph in which the edges were present before their removal. So,
\alg\ does not perform any action for explicit preservation of single-level
edges. To preserve the information that is lost due to the removal of
multi-level edges, \alg\ inserts additional even-topology vertices, together
with additional edges between the even-topology vertices to prepare the DAG for
the compression. The insertion of additional vertices and edges for preserving
the information of a removed DAG multi-level edge $e = (u, v)$ is discussed
below along with an example given in Figure~\ref{fig:dag-eg}. In this figure,
the topological levels are mentioned in rectangular boxes. On the left side we
show the original graph, and on the right side we show the modified graph which
preserves information that is lost due to compression. 

There are four possible cases for an edges $(u, v)$ that is being removed due
to topological compression.

\noindent \textbf{Case 1:} ($topo(u)$ is odd  and  $topo(v)$ is even).
Compression removes the vertex $u$, so we add a \textbf{fictitious} vertex $u'$
such that $topo(u')= topo(u)+1$. Then we remove the multi-level edge $(u,v)$
and replace it with with two edges $(u,u')$ and $(u',v)$. Since topological
level number of both $u'$ and $v$ are even, the topological compression does
not delete the edge $(u',v)$. For example, consider the multi-level edge
$(b,l)$ in figure~\ref{fig:dag-eg}(a), $topo(b) = 1$ (odd), and $topo(l) = 4$
(even). In the modified graph Figure~\ref{fig:dag-eg}(b) this edge is replaced
by two edges $(b,b')$ and $(b',l)$, where $b'$ is the fictitious node.\\

\noindent
\textbf{Case 2:} ($topo(u)$ is even  and  $topo(v)$ is odd). This case is
symmetric to Case 1 as compression removes $v$ instead of $u$. We use a similar
approach like Case 1 to handle this case. We create $v_1$, a copy of the vertex $v$
such that $topo(v_1)=topo(v)-1$ and replace the multi-level edge
$(u,v)$ with two edges $(u,v_{1})$ and $(v_{1},v)$. To distinguish 
the vertices added in these two cases, the newly added vertex is called 
\textbf{fictitious} for Case 1, and it is called \textbf{copied} for Case 2.
The justification of such naming will be clarified in latter part of the text.
Example of Case 2 in Figure~\ref{fig:dag-eg}(a) is edge $(m,s)$, where $topo(m) = 4$ (even) 
and $topo(s) = 7$ (odd). In modified graph, we add copied node $s_{1}$
and replace the original edge with two edges shown in Figure~\ref{fig:dag-eg}(b).\\
	
\noindent
\textbf{Case 3:} ($topo(u)$ is odd  and  $topo(v)$ is odd). In this case we use
the combination of above two methods and add two new vertices $u'$ and $v_{1}$.
We set topological level numbering of new vertices as mentioned above.  Also we
replace multi level edge $(u,v)$ with three different edges $(u,u')$ ,
$(u',v_{1})$, and $(v_{1},v)$.  Multi level edge $(h,r)$
in Figure~\ref{fig:dag-eg}(a) is an example of this case.  As shown in
Figure~\ref{fig:dag-eg}(b), we add two new vertices $h'$ and $r_{1}$ and three
new edges, $(h,h')$, $(h',r_{1})$, and  $(r_{1},r)$ after deleting the original
edge $(h,r)$. Note that, if $topo(u) = topo(v) - 2$, $topo(u')
= topo(v_{1})$. In this case, we treat it as Case 1 by adding only $u'$ (but
not $v_{1}$) and following the Case 1. It generates a single-level edge
$(u',v)$, which we do not need to handle explicitly.\\

\noindent
\textbf{Case 4:} ($topo(u)$ is even and  $topo(v)$ is even). This is the
easiest case as both $u$ and $v$ are not removed by the compression process and
we do not make any change in the graph.  Also note that the changes in the
above three cases convert those cases into this Case 4. For example, applying
Case 1 for edge $(b,l)$ in Figure~\ref{fig:dag-eg} creates new multi-edge
$(b',l)$ which is an occurrence of Case 4. Similarly Case 2 creates the Case 4
multi-edge $(m, s_{1})$.~$\blacksquare$\\

\begin{figure}
\centering
\includegraphics[width=90mm,keepaspectratio]{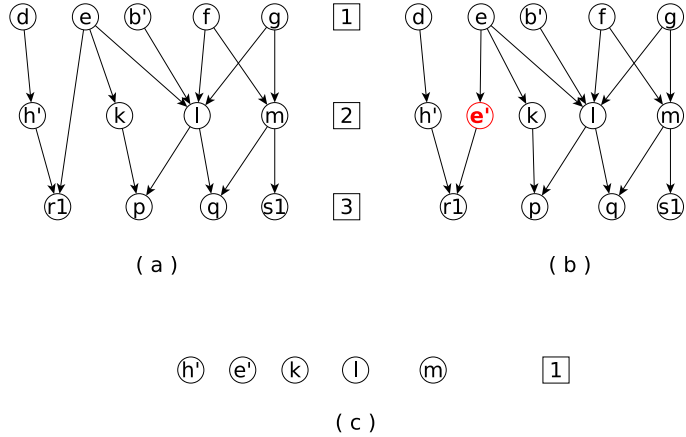}
\caption{(a) 1-Compressed Graph $G^{1}$,  (b) Modified 1-compressed Graph $G^{1}_{m}$,  (c) 2-Compressed Graph $G^{2}$}
\label{fig:comgraph}
\end{figure}

\noindent
\textbf{Dummy edges data structure:} 
We described earlier, we do not need to handle single-level edges separately.
However if two continuous single-level edges are removed, we still need to
maintain the logical connection between the even-topology vertices. For
example, in Figure~\ref{fig:dag-eg}(a) edges $(e,i) ,(i,k)$, and $(i,l)$ are
single-level edges which will be deleted after the first compression iteration
because $topo(i) = 3$. Now, information of logical (indirect) connection
between $e$ to $k$ and $l$ needs to be maintained, because all three vertices
will exist after the compression. To handle this, we add new \textbf{dummy
edges} $(e,k)$ and $(e,l)$; dummy edges are shown as dotted lines in
Figure~\ref{fig:dag-eg}(b). Note that, for any dummy edge $(u,v)$, $topo(v) -
topo(u) = 2$ in the current DAG and the edges for which the node-topology
difference is higher than 2 are handled by the above 4 multi-level edge cases.
For same start and end nodes, if there are multiple dummy edges, \alg\
considers edge with the smallest distance.  To find the dummy edges, we scan
through all odd-topology vertices and find their single-level incoming and
outgoing edges.  We store all these dummy edges along with the corresponding
\textit{distance} value in a list called $DummyEdges$ as shown in
Figure~\ref{fig:dag-eg}, which we use during the index generation step.
For example dummy edge $(d,h')$ has a distance 2 in the Figure~\ref{fig:dag-eg},
then $[(d,h'),2]$ is stored in $DummyEdges$.

At each compression iteration, we first obtain a modified graph, with
fictitious vertices, copied vertices, and dummy edges and then apply
compression to obtain the compressed graph of the subsequent iteration. The
fictitious vertices, copied vertices, and dummy edges of the modified graph in
earlier iteration become regular vertices and edges of the compressed graph in
subsequent iteration. The above modification and compression proceeds
iteratively until we reach $t$-compressed graph, $G^t$, for which the
topological level number is 1. We use $G_m$ to denote the modified uncompressed
graph, $G^1_m$ to denote the modified 1-compressed graph, $G^{2}_{m}$ to denote
the modified 2-compressed graph and so on. For example,
Figure~\ref{fig:comgraph}(a) shows $G^1$ which is obtained by compressing the
modified graph $G_m$ in Figure~\ref{fig:dag-eg}(b).
Figure~\ref{fig:comgraph}(b) shows $G^1_m$, the modified 1-compressed graph,
and Figure~\ref{fig:comgraph}(c) shows $G^2$, the 2-compressed graph. We refer the
set of all modified compressed graphs as $G^*_m$, i.e. $\lbrace$
$G^t_m,...,G^2_m,G^1_m$ $\rbrace$. 

\subsection{Index generation}
\label{sec:index_gen}
\alg's index data structure is represented as a table of \textit{key-value}
pairs. For each \textit{key}, (vertex) $v$ of the input graph, the \textit{value}
contains two lists: (i) outgoing index value $I^{out}_{v}$ , which
stores shortest distances from $v$ to a set of vertices reachable from $v$; and
(ii) incoming index value $I^{in}_{v}$, which stores shortest
distances between $v$ and a set of vertices that can reach $v$. Both the lists
contain a collection of tuples, \texttt{$\langle vertex\_id, distance
\rangle$}, where \textit{vertex\_id} is the id of a vertex other than $v$, and
\textit{distance} is the corresponding shortest path distance between $v$ and
that vertex. 

\begin{algorithm}
\renewcommand{\algorithmicrequire}{\textbf{Input:}}
\renewcommand{\algorithmicensure}{\textbf{Output:}}
\caption{Outgoing Index \textit{value} Generation}
\label{alg1}
\begin{algorithmic}[1]
\REQUIRE $G^{*}_{m}$ (set of modified graphs), DummyEdges (set of dummy edges and corresponding distance)
\ENSURE $I^{out}_{*}$ (set of out going indexes for all nodes)
\FORALL {$G^{curr}_{m} \in \lbrace G^{topo(G)}_{m}, ... ,G^{1}_{m},G_{m}  \rbrace $}
  \STATE $O$ = $\lbrace u \in V(G^{curr}_{m})  | topo(u) = odd\_number \rbrace$
  \FORALL {$v \in O$}
   	\STATE $org\_v$ = \proc{GetOriginal$(v)$} 
    \FORALL {$(v,w) \in {E(G^{curr}_{m})}$}
      \STATE $org\_w$ = \proc{GetOriginal$(w)$}
      \IF {$org\_v == org\_w$}
        \STATE Continue
      \ENDIF
      \IF {$(v,w) == Dummy\_Edge$}
        \STATE distance = \proc{GetDummyDistance$(DummyEdges,v,w)$} 
      \ENDIF
      \IF {$w == fictitious\_vertex$}
        \STATE distance = distance - 1
      \ENDIF
      \STATE \proc{RecursiveInsert$(I^{out}_{org\_v},org\_w,distance,out)$} 
    \ENDFOR
  \ENDFOR        
\ENDFOR
\end{algorithmic}
\end{algorithm} 

At the beginning of the indexing step, for each vertex $v$, \alg\ initializes
$I^{out}_{v}$ and $I^{in}_{v}$ with an empty set. It generates index from
$G^t_m$ and repeats the process in reverse order of graph compression i.e. from
graph $G^t_{m}$ to $G_{m}$. In $i$'th iteration of index building, it uses
$G^{t-i}_m$ and inserts a set of tuples in $I^{out}_{v}$ and $I^{in}_{v}$, only
if $v$ is an odd-topology vertex in $G^{t-i}_{m}$.  Thus, during the first
iteration, for every odd-topology vertex $v$ of $G^{t-1}_m$, for an incoming
edge $(u, v)$ \alg\ first checks whether $(u, v)$ is in $DummyEdges$ data
structure, if so, it inserts $\langle u, d\rangle$ in $I^{in}_{v}$, where the
distance $d$ value is obtained from the $DummyEdges$ data structure. Otherwise,
it inserts $\langle u, 1\rangle$ in $I^{in}_{v}$. Similarly, for an outgoing
edge $(v, w)$ \alg\ inserts $\langle w, d\rangle$ in $I^{out}_{v}$, if $(v,w)$
is in $DummyEdges$, otherwise it inserts $\langle w, 1\rangle$ in
$I^{out}_{v}$.  \alg\ also inserts (Line 16 in Algorithm \ref{alg1}) elements
of $I^{in}_{u}$ and $I^{out}_{w}$ into $I^{in}_{v}$ and $I^{out}_{v}$,
respectively, using recursive calls.

Algorithm~\ref{alg1} shows the pseudo-code of the index generation procedure
for outgoing index \textit{values} only. An identical piece of code can be used
for generating incoming index \textit{values} also, but for that we need to
exchange the roles of fictitious and copied vertex, and change the $I^{out}_*$
with $I^{in}_*$ in Line 5-21 (more discussion on this is forthcoming). 

As shown in Line 2 of Algorithm~\ref{alg1}, \alg\ first collects all
odd-topology vertices in variable $O$ and builds out-indexes for each of these
vertices using outgoing edges from these vertices (the edge $(v, w)$ in Line 5
of Algorithm~\ref{alg1}). Note that, vertices $v$ and $w$ in $G^{curr}_m$ can
be fictitious or copied vertex; \alg\ uses the subroutine {\sc GetOriginal()}
which returns original vertex corresponding to any fictitious or copied vertex,
if necessary (Line 4 and 6).  Using the data structure $DummyEdges$ (discussed
in section~\ref{sec:tc}), it first checks whether the edge $(v, w)$ is a dummy
edge (Line 10); if so, it obtains the actual \textit{distance} from the data
structure. In case the end-vertex $w$ is a fictitious vertex, \alg\ decrements
the distance value by 1 (Line 14), because for each fictitious vertex, an extra
edge with distance 1 is added from the original vertex to the fictitious vertex
which has increased the distance value by one. For instance, in the graph in
Figure~\ref{fig:dag-eg}, the actual distance from $a$ to $h$ is 2, but the
fictitious vertex $h'$ records the distance to be $3$, which should be
corrected. On the other hand, if $w$ is a copied vertex, \alg\ does not make
this subtraction, because when a copied vertex is used as destination instead
of the original vertex, the distance between the source vertex and the copied
vertex correctly reflects the actual distance. For an example, in the same
graph, the distance between $e$ and $r$ is 1; when we use the copied vertex
$r1$ instead of $r$, distance between $e$ and $r1$ is recorded as 1, which is
the correct distance between $e$ and $r$; so no distance correction is needed
during the out index building when the destination vertex is a copied vertex.
This is the reason why we make a distinction between the fictitious
vertices and the copied vertices. 

Finallly note that, after generating indexes for each vertex there may be
multiple entries for some vertices; from these multiple entries we need to get
the smallest value (entry) and remove others.  For building incoming index
\textit{values}, \alg\ subtracts 1 for a copied vertex, but does not subtract 1
for a fictitious vertex, as the roles of start and end vertices flip for the
incoming index \textit{values}. Below, we give a complete index building
example using the vertex $a$ of the graph in Figure~\ref{fig:dag-eg}.\\

\begin{table}
\centering
\caption{Intermediate index generated from the DAG in Figure~\ref{fig:comgraph}(b)}
\label{Table:IntermediateIndex}
\begin{tabular}{c c|c c} 
\toprule
$\textbf{\textit{key}}$ & Out Index \textit{value} & $\textbf{\textit{key}}$ & In Index \textit{value}\\ 
\midrule

$\textbf{b}$ & $\lbrace \langle l,1\rangle \rbrace$ & $\textbf{p}$ & $\lbrace \langle k,2\rangle,\langle l,2\rangle \rbrace $\\ 

$\textbf{d}$ & $\lbrace \langle h,1\rangle \rbrace$ & $\textbf{q}$ & $\lbrace \langle m,1\rangle , \langle l,2\rangle \rbrace$\\ 

$\textbf{e}$ & $\lbrace \langle k,2\rangle ,\langle l,2\rangle \rbrace$ & $\textbf{r}$ & $\lbrace \langle h,1\rangle ,\langle e,1\rangle \rbrace$\\ 

$\textbf{f}$ & $\lbrace \langle l,2\rangle ,\langle m,2\rangle \rbrace$ & $\textbf{s}$ & $\lbrace \langle m,1\rangle \rbrace$\\ 

$\textbf{g}$ & $\lbrace \langle l,2\rangle ,\langle m,2\rangle \rbrace$ & {} & {}\\

\bottomrule

\end{tabular}
\end{table}

\begin{table*}
\centering
\caption{Index for the DAG in Figure~\ref{fig:dag-eg}}
\label{Table:Index-result}
\begin{tabular}{c c c} 
\toprule
$\textbf{\textit{key}}$ & Out Index \textit{value} & In Index \textit{value} \\ 
\midrule

$\textbf{a}$ & $\lbrace \langle d,1\rangle,\langle e,1\rangle,\langle h,2\rangle,\langle k,3\rangle,\langle l,3\rangle \rbrace$ & $\emptyset$ \\ 

$\textbf{b}$ & $\lbrace \langle l,1\rangle,\langle f,1\rangle,\langle m,3\rangle \rbrace$ & $\emptyset$ \\ 

$\textbf{c}$ & $\lbrace \langle f,1\rangle,\langle g,1\rangle,\langle l,3\rangle,\langle m,3\rangle \rbrace$ & $\emptyset$\\ 

$\textbf{d}$ & $\lbrace \langle h,1\rangle \rbrace$ & $\emptyset$ \\ 

$\textbf{e}$ & $\lbrace \langle k,2\rangle ,\langle l,2\rangle \rbrace$ & $\emptyset$ \\

$\textbf{f}$ & $\lbrace \langle l,2\rangle ,\langle m,2\rangle \rbrace$ & $\emptyset$ \\

$\textbf{g}$ & $\lbrace \langle l,2\rangle ,\langle m,2\rangle \rbrace$ & $\emptyset$ \\

$\textbf{h}$ & $\emptyset$ & $\lbrace \langle d,1\rangle \rbrace$ \\

$\textbf{i}$ & $\lbrace \langle k,1\rangle,\langle l,1\rangle \rbrace$ & $\lbrace \langle e,1\rangle \rbrace$ \\

$\textbf{j}$ & $\lbrace \langle l,1\rangle,\langle m,1\rangle \rbrace$ & $\lbrace \langle f,1\rangle,\langle g,1\rangle \rbrace$ \\

$\textbf{n}$ & $\lbrace \langle p,1\rangle \rbrace$ & $\lbrace \langle k,1\rangle \rbrace$ \\ 

$\textbf{o}$ & $\lbrace \langle p,1\rangle \rbrace$ & $\lbrace \langle l,1\rangle \rbrace$ \\ 

$\textbf{p}$ & $\emptyset$ & $\lbrace \langle k,2\rangle ,\langle l,2\rangle \rbrace$ \\ 

$\textbf{q}$ & $\emptyset$ & $\lbrace \langle m,1\rangle, \langle l,2\rangle \rbrace$ \\

$\textbf{r}$ & $\emptyset$ & $\lbrace \langle h,1\rangle,\langle e,1),\langle p,1\rangle,\langle k,3\rangle,\langle l,3\rangle\rbrace$ \\

$\textbf{s}$ & $\emptyset$ & $\lbrace \langle m,1\rangle,\langle p,1\rangle,\langle k,3 \rangle ,\langle l,3 \rangle,\langle q,1 \rangle \rbrace$ \\ 
\bottomrule
\end{tabular}
\end{table*}

\noindent {\bf Example:} We want to find the outgoing index (\textit{value})
for vertex $a$ (\textit{key}) of the graph $G$ in Figure~\ref{fig:dag-eg}(a).
$topo(G)=2$, so we start building index using the graph $G^1_m$, which is shown
in Figure~\ref{fig:comgraph}(b). In the first iteration, \alg\ builds
$I^{out}_{d} = \lbrace \langle h,1 \rangle \rbrace$; the distance value of $1$
comes as follows: \alg\ uses distance of dummy edge $(d,h')$ that is 2
(Figure~\ref{fig:dag-eg}-II) and then it replaces the fictitious vertex $h'$
with $h$ and obtains a distance of 1 by subtracting 1 from 2 (Line 14).  It
also inserts the following entries under the {\em key} $e$; i.e., $I^{out}_{e}
= \lbrace \langle k,2 \rangle,\langle l,2 \rangle \rbrace$.  The resulting
indexes after this iteration is presented in
Table~\ref{Table:IntermediateIndex}; incoming index \textit{values} for keys
$b,d,e,f,g$ are empty (not presented in the table) and similarly outgoing index
\textit{values} for keys $p,q,r,s$ are empty.  For next iteration considering
$G_m$, \alg\ inserts $\langle d,1 \rangle$ in $I^{out}_{a}$; using recursive
calls of algorithm~\ref{alg2} (Line 11), this function also inserts $\langle
h,2 \rangle$ in $I^{out}_{a}$, recursion stops at $h$ because $I^{out}_h$ is
empty (Line 6).  Similarly, $\langle e,1 \rangle$ and $\langle k,3
\rangle,\langle l,3 \rangle$ are inserted in $I^{out}_a$ recursively from
$I^{out}_{e}$.  At the end of the algorithm~\ref{alg1} we remove duplicate
entries from indexes.  For example, incoming index for key $s$ has two entries
for vertex $m$, $\langle m,1 \rangle$ and $\langle m,2\rangle$, one
corresponding to edge $(m,s1)$ in $G^{1}_{m}$ and the other is a recursive
result from $q$ to $s$ in $G_{m}$. \alg\ considers $\langle m,1\rangle$ and
discards the other entry from $I^{in}_{s}$. For the graph in
Figure~\ref{fig:dag-eg}(a), corresponding indexes are presented in
Table~\ref{Table:Index-result}.

\begin{algorithm}
\renewcommand{\algorithmicrequire}{\textbf{Input:}}
\renewcommand{\algorithmicensure}{\textbf{Output:}}
\caption{\proc{RecursiveInsert$(I^{io}_{v},a,distance,in\_or\_out)$}}
\label{alg2}
\begin{algorithmic}[1]
\IF {$in\_or\_out == in$}
  \STATE $I^{io}_{a} = I^{in}_{a}$
\ELSE
  \STATE $I^{io}_{a} = I^{out}_{a}$
\ENDIF
\IF{$I^{io}_{a} == \emptyset$}
  \STATE $add\_tuple(I^{io}_{v},a,distance)$
\ELSE
  \STATE $add\_tuple(I^{io}_{v},a,distance)$
  \FORALL {$(x,dist) \in I^{io}_{a} $}
    \STATE \proc{RecursiveInsert$(I^{io}_{v},$ $x,$ $distance$ $+$ $dist,$ $in\_or\_out)$}
  \ENDFOR
\ENDIF
\end{algorithmic}
\end{algorithm}

\subsection{Index for weighted graph}

For weighted graph, \alg\ makes some minor changes in the above algorithm.
First, distance values are stored both in the indexes and in the $DummyEdge$
data structure. Many of these distances are implicitly 1 for unweighted graph,
which is not true for weighted graph, so, for the latter \alg\ stores the
distance explicitly. Also, it ensures that the distance value between
fictitious (or copied) vertices and an original vertex is one, so that
the Algorithm~\ref{alg1} works as it is.

\subsection{Query processing}
\label{sec:query_pro}
For query processing, \alg\ uses the distance indexes that is built during the
indexing stage. For a given distance query from $u$ to $v$, i.e. to compute
$\delta(u,v)$, \alg\ intersects outgoing index \textit{value} of \textit{key}
$u$ i.e. $I^{out}_{u}$ and incoming index \textit{value} of \textit{key} $v$
i.e. $I^{in}_{v}$ and finds common \textit{vertex\_id} in $I^{out}_{u}$ and
$I^{in}_{v}$, along with the \textit{distance} values.  To cover the cases,
when $v$ is in the outgoing index {\em value} of $u$, or $u$ is in the incoming
index {\em value} of $v$, the tuples  $\langle u,0 \rangle$ and $\langle v,0
\rangle$ are also added in $I^{out}_u$ and $I^{in}_v$ respectively and then the
intersection set of the indexes is found.  If the intersection set size is 0, there is
no path from $u$ to $v$ and hence the distance is infinity.  Otherwise, the
distance is simply the sum of the \textit{distances} from $u$ to
\textit{vertex\_id} and \textit{vertex\_id} to $v$. If multiple paths exist, we
take the one that has the smallest distance value.\\
 
\noindent {\bf Example:} We want to find $\delta(a,s)$ in
Figure~\ref{fig:dag-eg}. From table~\ref{Table:Index-result}, $I^{out}_a \cap
I^{in}_s=\{k,l\}$. Now, we need to sum up the corresponding \textit{distance}
values, that gives $\{\langle k,6\rangle,\langle l,6\rangle \}$. Now we need to
find smallest distance value; in this case both the values are same, hence we
can provide any one as a result.

\subsection{Theoretical proofs for correctness}
\label{sec:proof}

In this section, we prove the correctness of \alg, through the claim that
\alg's index is based on 2-hop covers of the shortest distance in a graph and 
method described in Section~\ref{sec:query_pro} gives correct shortest distance value. 
For shortest path, such a cover is a collection $S$ of shortest paths such that
for every two vertices $u$ and $v$, there is a shortest path from $u$ to
$v$ that is a concatenation of atmost two paths from $S$.~\cite{2Hop:Cohen:2002}.
That is, shortest path from $u$ to $v$ is stored in $S$ or there is an 
intermediate node $x$ such that shortest paths from $u$ to $x$ and from $x$ to 
$v$ are stored in $S$. For \alg's index also, the shortest distance from any node $u$
to node $v$ is the 2-hop cover such that the index itself has shortest distance value 
from $u$ to $v$ or there is an intermediate node $x$ which would be present in 
both $I^{out}_{u}$ and  $I^{in}_{v}$. \\

\noindent \textbf{Example:} In DAG $G$ in Figure~\ref{fig:dag-eg}(a) to find distance 
from $a$ to $s$, we need to check the outgoing index value for vertex $a$ and 
the incoming index value for vertex $s$ in Table~\ref{Table:Index-result}. This 
gives us two possible shortest paths passing through intermediate node
$k$ or $l$, because distance in both cases is same. Thus, there can be multiple
shortest paths however, atmost one intermediate node in the index.

In the Theorem~\ref{thrm:main}, we try to identify the topological layer of an intermediate node $x$. 
We identify a unique topological level for each pair of $u$ and $v$, which tells
there is atmost one intermediate node in a shortest path from $u$ to $v$ because  
in DAG there cannot be a directed edge within topological layer. We begin with the 
following lemmas, which will be useful for constructing the proof of the theorem.

\begin{lemma} In $G_m$, if a node $u$ is at topological level $2^i$, 
it will be at topological level $1$ in $G^{i}$.
\label{lemma1}	
\end{lemma}
\begin{proof}
\alg\ compression method removes all odd-topology nodes and carries over nodes
from the even topological levels to the next compression iteration. Thus any
node from an even topological level $2x$ in some compressed graph will be at
topological level $x$ in the compressed graph of next iteration. Say, the node
$u$ is at topological level $2^i$ in $G_m$, then it will be at topology level
$2^{i-1}$ in $G^1$. Since $2^{i-1}$ is also even, no fictitious or copied
vertex will be added for $u$, and in $G^{1}_m$, it will remain at $2^{i-1}$ level.
In the next compression iteration, $u$ will simply be moved to $2^{i-2}$ level
in $G^2$ and so on. Hence, it will be at level $2^{i-i}=2^0=1$ level in $G^{i}$ graph.
\end{proof}
\noindent {\bf Example} In the graph $G_m$ shown in Figure~\ref{fig:dag-eg}(b), the node
$d$ is at topological level $2$ and the node $k$ is at topological level $4$. In $G^1$
shown in Figure~\ref{fig:comgraph}(a) the node $d$ is at topological level $1$;
similarly, in $G^2$ shown in Figure~\ref{fig:comgraph}(c), the node $k$ is at
topological level $1$.

\begin{lemma}
\label{lemma2}
In \alg's index, for all keys, the \textit{values} contain vertices, which are only from 
even topological level in the modified DAG $G_m$.
\end{lemma}
\begin{proof}
As per Line 2 of Algorithm~\ref{alg1}, \alg's index \textit{keys} are nodes
from only an odd topological level, and the {\it values} of index are built
using the incident edges of those key nodes. In the modified graph $G_m$, all
the edges from/to an odd topology vertex connects with an even topology vertex,
through the use of fictitious/copied nodes (if needed). Hence, if any node in 
DAG $G_m$ is at an odd topological level, it cannot be included as an index 
\textit{value}. Additionally when we compress $G^{i}_m$ to get $G^{i+1}$, we 
only include nodes from even topological levels, hence nodes from odd levels 
will never be included as a \textit{value} for index building at compressed levels also.
\end{proof}
\noindent {\bf Example:}
See the completely built index of the graph $G$ in Figure~\ref{fig:dag-eg}(a)
as shown in Table~\ref{Table:Index-result}. The nodes that appear as \textit{values} 
are $\{d,e,b'(b),f,g,h'(h),k,l,m,r_1(r),p,q,s_1(s)\}$. All of these are 
from the even topology nodes in $G_m$ as shown in Figure~\ref{fig:dag-eg}(b).\\

\begin{figure}
\centering
\includegraphics[width=50mm,keepaspectratio]{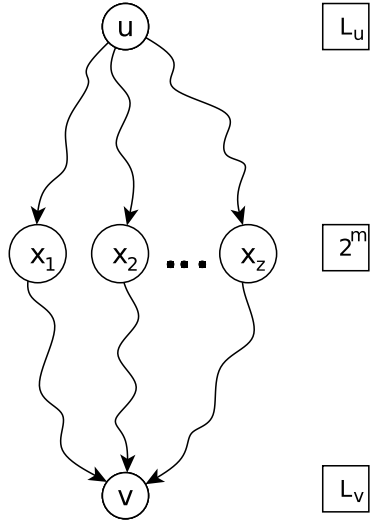}
\caption{Shortest path from $u$ to $v$ passing through $x$}
\end{figure} 

\begin{theorem}
\label{thrm:main}
For finding shortest distance from $u$ to $v$, assume that $u$ has 
topological level number $L_u$ and $v$ has topological level number 
$L_v$ in $G_m$. We define
\begin{equation}
\label{eq1}
    n = \argmax_{i} ( L_u \leq 2^i \leq  L_v)
\end{equation}

Now, if there is a shortest path from $u$ to $v$, for each shortest path, 
exclusively, one of the following is true.\\

Case 1: No intermediate node $x$ i.e. $I^{out}_{u}$ includes $v$ or 
$I^{in}_{v}$ includes $u$. \\

Case 2: There is an intermediate node $x$, and \[ topo(x) = 2^n + C \]

\hspace{13mm} for some constant offset C.
\end{theorem}

\begin{proof}
We prove this theorem using mathematical induction on $n$.\\

\textbf{Base case:} $n=1$. If there is a direct edge from $u$ to $v$ 
then \textit{case 1} is true because if $L_u = 2$ then $I^{in}_{v}$ includes $u$ 
or if $L_v =2 $ then $I^{out}_{u}$ includes $v$. If there is an intermediate node 
$x$ then $L_u = 1$ and $L_v=3$, hence $topo(x) = 2$ that shows \textit{case 2} is true. 
From Lemma~\ref{lemma2} both $I^{out}_{u}$ and $I^{in}_{v}$ include the node $x$ 
if there is a path from $u$ to $v$. In this case constant offset \textit{C} would be zero.\\

\textbf{Induction hypothesis:} Here we assume that for $n = d$ given theorem is true.\\

\textbf{Induction step:} We want to prove, for $n=d+1$ given theorem is true. 

If there is no intermediate node $x$ then \textit{case 1} is true. Hence, we discuss the 
only scenario where there is an intermediate node $x$ and we want to show that $x$ is 
in both $I^{out}_{u}$ and $I^{in}_{v}$. We sub-divide the proof for zero and non-zero 
values of constant offset \textit{C}.\\

\noindent \emph{Constant offset \textit{C} is zero:}

If there are $2^{d+1}$ levels, then compression step would be conducted at least
one more time than $2^d$ levels. From Lemma~\ref{lemma1} at the $d$'th step of 
compression, nodes at topological level $2^d$ in graph $G_m$ are at $1^{st}$ 
topological level in $G^{d}$ and nodes from topological level $2^{d+1}$ would be 
at $2^{nd}$ topological level.

Hence, \alg\ will build index for \textit{keys} (nodes) from topological level
$2^d$ in graph $G_m$, and those index \textit{values} include nodes from
topological level $2^{d+1}$. As our method recursively includes already built
index \textit{values}, the nodes from topological level $2^{d+1}$ would be
recursively included to corresponding outgoing index \textit{values} for
\textit{keys} at lower compression levels. Hence, if $topo(x) = 2^{d+1}$ then it
is present in outgoing index value of $u$.

The similar argument works for incoming index of $v$.\\

\noindent \emph{Constant offset \textit{C} is non-zero:}

If we cannot find $n$ that satisfies equation~\ref{eq1} then constant offset 
\textit{C} is non-zero. In this case offset can be calculated as :
\begin{equation}
\label{eq2}
	\begin{split}
		C &= 2^{n_{low}} \\
		where, \qquad 
	    n_{low} &= \argmax_{i}  ( 2^i <  L_u )
    \end{split}
\end{equation}

Now, we define modified topological level number of $u$ is $L^m_u$, where 
$ L^m_u = L_u - C$ and similarly modified topological level number of $v$  
is $L^m_v = L_v - C$. We use $L^m_u$ and $L^m_v$ in equation~\ref{eq1} to get 
$n$ \[ n = \argmax_{i} ( L^m_u \leq 2^i \leq  L^m_v) \]

Now, with $topo(x) = 2^{n} + C$, argument works similarly as zero offset.

\end{proof}

Example of non-zero offset \textit{C}: In Figure~\ref{fig:dag-eg}(b), we want to know
the shortest distance from $n$ to $r$ where corresponding topological levels
are $L_n = 5$ and $L_r = 7$ respectively. For this, we can not find any $n$
that satisfies the equation~\ref{eq1}. From equation~\ref{eq2}, we can
calculate $n_{low}=2$, using which modified topological levels $L^m_n
=1(5-2^2)$ and $L^m_r=3$ can be obtained. From $L^m_n$ and $L^m_r$ we get
$m=1$. Now, $2^1 + 2^2 = 6$ is the topological level of intermediate node $p$, 
which is present in both $I^{out}_{n}$ and $I^{in}_{r}$ (Table~\ref{Table:Index-result}). 
If we look carefully $L^m_n$ and $L^m_r$ is a base case in the mathematical
induction proof of Theorem~\ref{thrm:main}, and $L_n(5)$, $L_r(7)$ with offset \textit{C}
behave exactly the same as the base case.\\

\textbf{Note:} If there is no node from topological level $2^n$ in the shortest path from $u$
to $v$, then there must be one multilevel edge which skips that level. For a
node incident to that multilevel edge, at some step of the compression, we need
to prepare fictitious/copied node. That new fictitious/copied node works as a
node from topological level $2^n$ and will be included in both $I^{out}_u$ and
$I^{in}_v$. Thus, theorem works fine for this case.

For example, as depicted in Figure~\ref{fig:dag-eg}(b),
shortest path from $a$ to $r$ doesn't pass through any node from topological
level $4(2^2)$ in $G_m$, but it has a multilevel edge $(e,r_1)$. In  $G^2_m$ 
(Figure~\ref{fig:comgraph}(b)), this edge causes a fictitious node $e'$ at
topological level $2$ which is (logically) topological level $4$ in $G_m$. The
resulting index in Table~\ref{Table:Index-result} shows that, the node $e'(e)$
is included as a \textit{value} in the incoming index of $r$ ($I^{in}_r$) and
also included in  $I^{out}_a$.

\begin{figure*}[t]
\centering
    \begin{subfigure}[h]{0.5\textwidth}
	\centering
		\includegraphics[width=70mm,keepaspectratio]{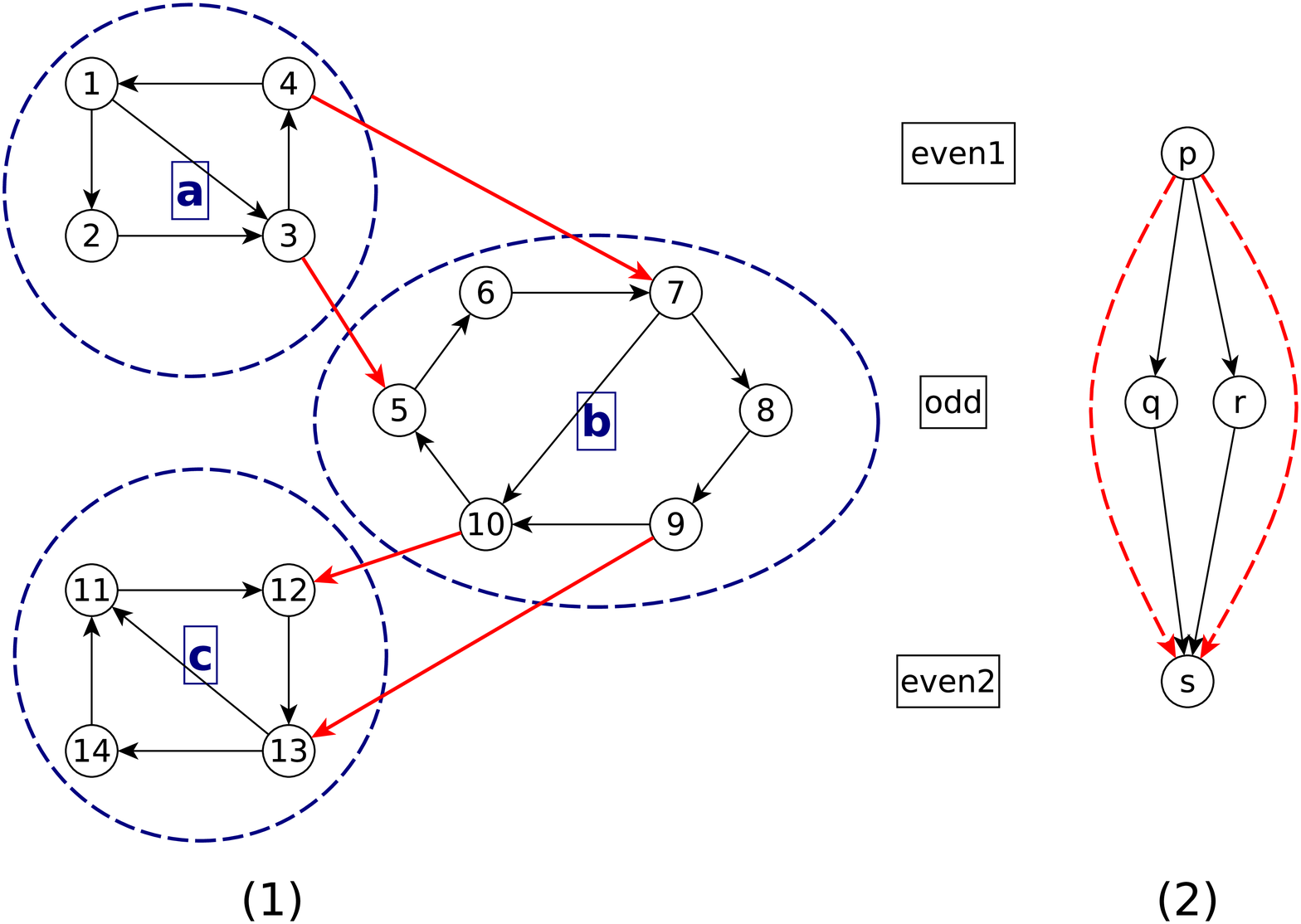}
		\caption{(1)Example for Distance within Middle DAG node of $G_{d}$ and (2)Example of multiple dummy edges in modified $G_{d}$}
		\label{fig:scc}
    \end{subfigure}%
    ~
    \begin{subfigure}[h]{0.5\textwidth}
	\centering

		\includegraphics[width=60mm,keepaspectratio]{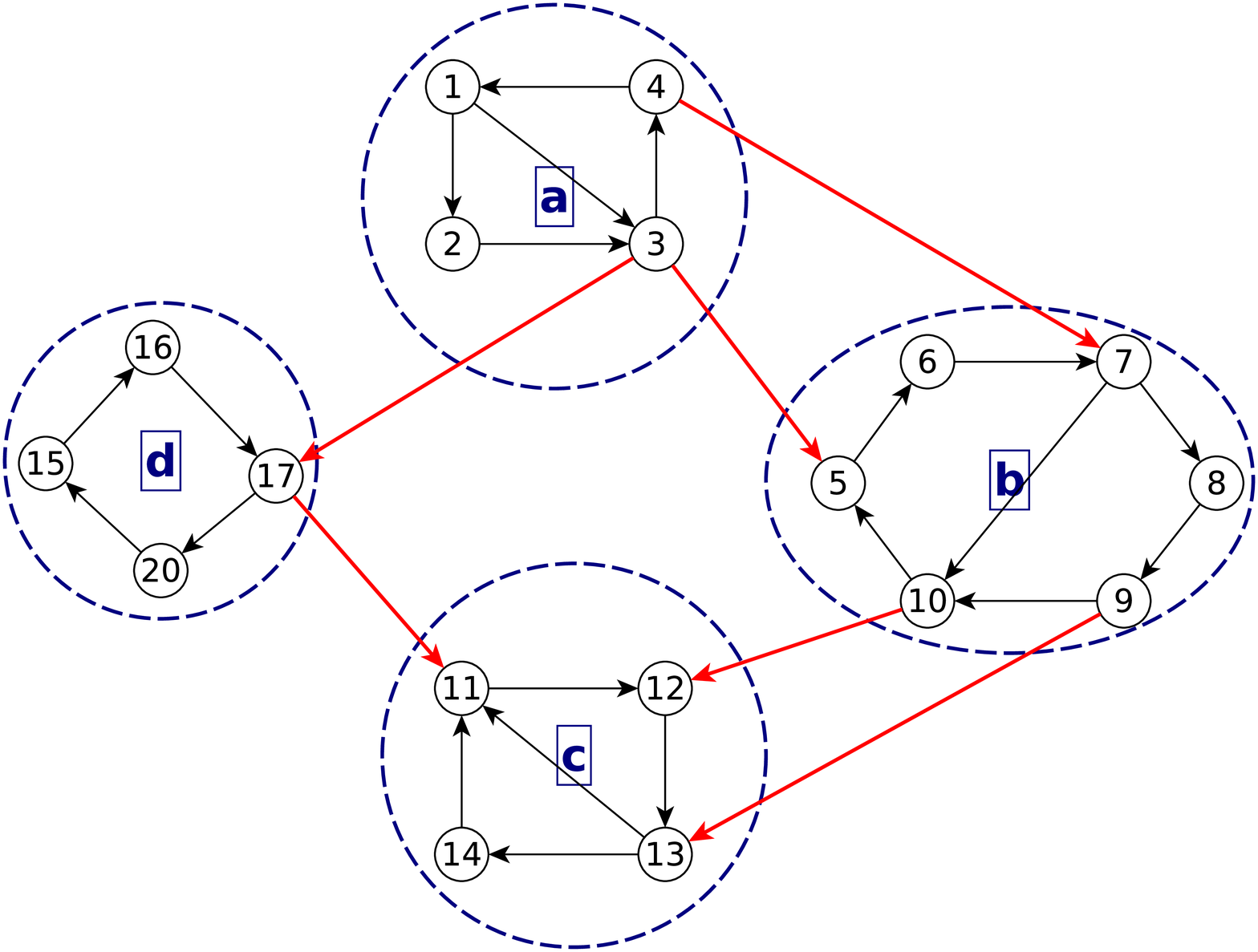}
		\caption{Example of merging multiple dummy DAG edges}
		\label{Figure:merging dummy}
    \end{subfigure}
    \vspace{-0.1in}
	\caption{Dummy edge handling}	

    \hfill

\end{figure*}

\section{Indexing for general directed graph}
\label{sec:generalized-directed-graph}

Any directed graph $G$ can be converted to a Directed Acyclic Graph (DAG)
$G_d$, by considering each strongly connected component (SCC) of $G$ as a node
of $G_d$. Thus in DAG, the edges within a SCC are collapsed within the
corresponding node. However, if an edge in $G$ connects two vertices from two
distinct SCCs, in $G_d$ those SCCs are connected by a DAG edge. To build the
shortest path index for a general directed graph, \alg\ first uses Tarjan's
algorithm \cite{DFS:Tarjan:1972} to convert $G$ to a DAG $G_d$ by finding all
SCCs of $G$. It also maintains a necessary data structure that keeps the
mapping from a DAG node to a set of graph vertices, and vice-versa. We call
this a parent-child mapping, i.e., a DAG node is the parent of graph vertices
which are part of the corresponding SCC. A DAG edge connects two vertices, one
from a distinct SCC. We call such vertices terminal vertices. A single DAG edge
between a pair of SCCs may encapsulate multiple paths (one edge or multiple
edges) of Graph $G$ such that the end vertices of these paths are terminal
vertices in those pair of SCCs. \alg's DAG edge data structure contains a set
of tuples, each representing one of these paths. A tuple has three elements:
node-id of start terminal vertex, node-id of end terminal vertex, and the
distance between these two vertices in $G$.  For example consider
Figure~\ref{fig:scc}(1), a DAG edge $(a,b)$ stores $\lbrace (3,5,1),(4, 7, 1),
(3,7,2), (4,5,3) \rbrace$, first two tuples represent a single-edge path, but
the last two represent multi-edge paths. 3, 4 are terminal vertices of DAG node
$a$, and 5,7 are terminal vertices of DAG node $b$.  For each tuple, the third
field stores the shortest distance between the pair of terminal vertices in the
first two fields of that tuple.  To compute distance between an arbitrary pair
of vertices within an SCC, \alg\ also pre-computes all-pair shortest paths
among all graph nodes belonging to a single SCC and store them in shortest path
index.  In real-life directed networks, the size of SCCs are generally not very large,
so storing all-pair distances within a SCC in the shortest path index is feasible.

\subsection{Distance for dummy edges:}
\label{sec:dummyDist}

For an unweighted DAG two consecutive edges yield a distance value of 2, but
for DAG which is a compressed representation of a general unweighted directed graph, 
two consecutive DAG edges may constitute an arbitrary distance value. This is due
to the fact that the shortest path may visit a large number of vertices which
are part of the start, middle, and end SCC. For an example, see
Figure~\ref{fig:scc}(1); in this figure, the rectangles ``even1", ``odd" and
``even2" are the topological level numbers; $a$, $b$ and $c$ are the DAG nodes 
(ellipses), and nodes with numeric ids are nodes of $G$.  Two consecutive DAG
edges are $(a, b)$ and $(b, c)$ connecting SCCs $a$, $b$ and SCCs $b$, $c$,
respectively. The shortest distance between $a$ and $c$ depends on the terminal
nodes of $a$ and $c$ that are being used.  If terminal node of $a$ is $3$ and
terminal node of $c$ is 13, the distance is 5, following the path
$3--4--7--10--12--13$. In this path, besides the distance 2 over the DAG edges,
there are three within-SCC edges, one in each of SCCs.  Thus, the total
distance for a dummy edge is the sum of (i) distance from a terminal vertex of
staring SCC to a terminal vertex in the middle SCC, (ii) distance between a
pair of terminal vertices in the middle SCC, and (iii) distance from a terminal
node in the middle SCC to  a terminal node in the end SCC. To account for this,
\alg\ computes the dummy edge distance by considering all possible combination
of terminal nodes in each SCCs.\\
\begin{table*}
\centering
\caption{Real world datasets and basic information}
\label{table:dataset}
\hspace{-0.5in}
\resizebox{1.1\textwidth}{!}{
\begin{tabular}{l c c c c c c c c c} 
\toprule
Name & $\vert$ V $\vert$ & $\vert$ E $\vert$ & AD & MD & $\vert$ $V_{DAG}$ $\vert$ & $\vert$ $E_{DAG}$ $\vert$ & $AD_{DAG}$ & $MD_{DAG}$ & Largest SCC\\ 
\midrule
AS\_Caida & 26,374 & 2,304,095 & 87.36 & 2,205,805 & 	26,358 & 48,958 & 1.86 & 2606 & 8\\ 

Email\_Eu & 265,214 & 420,045 & 1.58 & 7,636 & 231,000 & 223,004 & 0.97 & 168,815 & 34,203\\ 

Epinion & 49,289 & 487,183 & 9.88 & 2,631 & 16,264 & 16,497 & 1.01 & 15,789 & 32,417\\ 

Gnutella09 & 8,114 & 26,013 & 3.21 & 102 & 5,491 & 6,495 & 1.18 & 5,147 & 2,624\\ 

Gnutella31 & 62,586 & 147,892 & 2.36 & 95 & 48,438 & 55,349 & 1.14 & 43,928 & 14,149\\ 


WikiVote & 7,116 & 103,689 & 14.57 & 1,167 & 5,817 & 19,540 & 3.36 & 4869 & 1,300\\ 
\bottomrule
\end{tabular}
}
\end{table*}

\noindent {\bf Example:}  Say, \alg\ wants to find dummy DAG edge $a$ to $c$ which would
be set of tuple $\lbrace (3,12,*), (3,13,*), (4,12,*), (4,13,*) \rbrace $, (all
sources to all destinations) where $*$ represents the shortest \textit{distance}
values that it needs to find. To find the distance from node $3$ to $13$, it
finds the distance for all combinations of terminal nodes in the middle SCCs
and takes the minimum.  From starting SCC to middle SCC ($\delta(3,5) = 1$ and
$\delta(3,7) = 2$), within middle SCC ($\delta(5,9) = 4$, $\delta(5,10) = 3$,
$\delta(7,9) = 2$ and $\delta(7,10) = 1$) and finally from middle SCC to end
SCC ($\delta(9,13) = 1$, $\delta(10,13) = 2$).  In this example,
$\delta(3,7)+\delta(7,10)+\delta(10,13)$ gives the minimum value $5$ which
generates the tuple $(3,13,5)$.  Distance for all other tuples are also
calculated similarly.\\

\noindent \textbf{Multiple dummy edges:}
Another issue is, there could be multiple dummy edges having same starting and
ending DAG nodes as shown in Figures~\ref{fig:scc}(2).  If the original graph
itself is a DAG, \alg\ considers the dummy edge with the lowest
\textit{distance}.  But for converted DAG $G_d$ applying this solution is more
complex, because distance within middle SCC can be different for different
SCCs.  For this, we need to merge all possible tuples of all dummy edges, and
recalculate the distances by taking the minimum distance from the merged set of
tuples.

\noindent {\bf Example} in Figure~\ref{Figure:merging dummy} we extend the
example of Figure~\ref{fig:scc}(1) with one more DAG node $d$, which also
connects node $a$ to node $c$.  Dummy edge through the middle node $d$ is
$\lbrace(3,11,2), (4, 11, 4), (3,12,3), (4,12,5), (3,13,4), (4,13,6)\rbrace $,
and through the middle node $b$ is
$\lbrace(3,11,6), (4,11, 5), (3,12,4), (4,12,3), (3,13,5),(4,13,4)\rbrace $.  For
dummy edge $(a,c)$, we combine both the sets of tuples and obtain the smallest
distance.  Thus the final representation of dummy edge $(a,c)$ is the
following: $\lbrace(3,11,2), (4,11,4), (3,12,3), (4,12,3), (3,13,4),
(4,13,4)\rbrace $.

\subsection{Modification in index and query processing}
\label{sec:modified_query}
\begin{table}
\centering
\caption{Average Query Time for DAG ($\mu$s)}
\label{table:results}
\begin{tabular}{l c c c c} 
\toprule
{Name} &  \textbf{TopCom} & \textbf{IS-Label} & \textbf{Bi-Djk} &  \textbf{TreeMap \footnotemark[3]}\\ 
\midrule
AS\_Caida & 0.1036 & 0.2237 & 24.75 & 0.2471  \\ 
Email\_Eu & 0.1059 & 0.3865 & 1657.46 & 0.2674 \\ 
Epinion & 0.0360 & 0.2388 & 14.83 & 0.1722\\ 
Gnutella09 & 0.0345 & 0.3292 & 7.27 & 0.115 \\ 
Gnutella31 & 0.0752 & 0.2095 & 50.74 & 0.254\\ 
WikiVote & 0.1551 & 0.3494 & 43.11 & 0.2131\\ 
\bottomrule
\end{tabular}
\end{table}

\footnotetext[3]{Unweighted Graph results}

\alg\ Index for general graph stores the bidirectional mapping between the DAG
nodes and the vertices of the input graph. For every DAG edge (and also for DAG
dummy edges), it stores the set of tuple based representation that we have
discussed in the above subsection. It also stores all pair distance between
each of the vertices within an SCC. Above all, it prepares and stores the 2-hop
cover DAG indexes for the DAG representation of the input graph using the
methodologies that we discussed in Section~\ref{sec:method}.

For query processing, given a query $(u, v)$, \alg\ first identifies the
corresponding SSE nodes in the DAG using the bidirectional map. Say, these SSEs
are $s_u$ and $s_v$, respectively. If $s_u = s_v$, \alg\ simply uses the within
SSE all-pair index and return the distance between $u$ and $v$.  Otherwise, it
first finds the set of out-terminal SSE nodes of $s_u$ (say, $X$), and
in-terminal SSE nodes of $s_v$ (say, $Y$). Then it uses the 2-hop cover
indexing for finding the shortest path distance between each pair of
nodes---one from $X$, and the other from $Y$. It also considers within
SCC distances in three SCCs:
\textbf{Starting DAG node:} 
Distance from starting node of the query to start terminal node of the DAG node.
For example consider figure~\ref{fig:scc}(1), 
where query is $\delta(1,14)$. 
Now all outgoing edges from DAG $a$ are from nodes $3$ and $4$, 
hence we need to get distances from $1$ to $3$ and $4$ 
i.e. $\delta(1,3) = 1$ and $\delta(1,4) = 2$.\\
\textbf{Middle DAG node:} 
If there is a middle node (from the 2-hop cover index) then distance from incoming edge terminal 
node to outgoing edge terminal node 
within middle DAG node needs to be calculated.
In our example suppose $b$ is middle node,
then we need distance from $5$ to $9$ and $10$, 
i.e. $\delta(5,9)=4$, $\delta(5,10)=3$ and similar for node $7$.\\
\textbf{Ending DAG node:} 
Distance from end terminal node to ending node of the query within the DAG node.:w
That means in our example distance from $12$ and $13$ to $14$ 
i.e. $\delta(12,14)=2$ and $\delta(13,14)=1$.

Here again our task is to get minimum distance among all, and we use the
similar strategy as section~\ref{sec:dummyDist}, which is to minimize the
summation of above three distances along with edge distances.

As we can see, \alg's principle indexing process works with DAG and it performs 
well on real world datasets (Table~\ref{table:IS-cmp}). One reason is, real 
world complex graph becomes less complex when converted as DAG. For example, 
average degree of DAG $AD_{DAG}$ (Table~\ref{table:dataset}) are always smaller 
than $AD$, mostly an order of magnitude smaller and specifically for $AS\_Caida$ 
average degree reduced to $1.86$ from $87.36$, which is a decrease of almost two 
orders of magnitude. However because of DAG, we also need to handle a challenging 
task, i.e. maintaining distance information within SCC for each DAG node. To keep 
this information, most common ways are to maintain distance matrix or to keep set 
of edges and calculate distance at run time. Both of these methods have their own 
pros and cons; keeping distance matrix is the fastest access method but the size 
of SCC leads to space limitation i.e. we need space for $O(n^2)$ elements and for 
huge SCC this may be a notable problem. On the other hand keeping set of edges is 
a memory efficient way, however all distance finding algorithms are polynomial time 
in terms of $\vert V\vert$ and $\vert E\vert$; and for huge SCC, finding distance
between nodes at run time would be much slower. This represents a well known
phenomenon in Computer Science called space-time trade-off. Here for our task,
time gets priority over space, hence we selected first method, where we are
maintaining distance matrix for each DAG node. For large SCC, this may take high 
memory, however we observed that our index is still not very large for 
contemporary machine.
\begin{table*}
\centering
\caption{Average Query Time for General Graph ($\mu$s)}
\label{table:IS-cmp}
\begin{tabular}{l c c c|c c} 
\toprule
{Name} &  \textbf{TopCom} & \textbf{IS-Label} & \textbf{Bi-Djk} &  \textbf{TreeMap \footnotemark[3]} & \textbf{TopCom \footnotemark[4]} \\ 
\midrule
AS\_Caida &  0.26282 & 0.2229 & 25.9614 & 0.3873 & 0.26044 \\ 
Email\_Eu  & 12.4708 & 21.98527 & 1482.41 & 0.8102 & 10.45136 \\ 
Epinion &  34.6582 & 2114.02 & 3570.8667 & 5.727 & 33.1907 \\ 
Gnutella09 & 1.3405 & 8.0429 &  110.6521 & 1.6255 & 1.29084 \\ 
Gnutella31 & 2.46202 & 13.9999 & 299.423 & 5.804 & 2.44672 \\ 
WikiVote & 18.3593 &  23.72954 &183.4193 & 8.371 & 19.06744 \\ 
\bottomrule
\end{tabular}
\end{table*}

\footnotetext[4]{Unweighted Graph: Avg. over 5 times execution for 10K queries}

\subsection{Correctness revisited}
In this Section~\ref{sec:generalized-directed-graph}, we explain how to adopt the proposed 
indexing method for a general directed graph. For that, first we convert a general directed 
graph into DAG and then build index on the DAG. We describe the methods to maintain 
the information at both steps.

We should be able to calculate shortest distance from one node to any other node within the 
same DAG node. When we convert a general directed graph to DAG, we maintain this information by 
creating appropriate data structures during conversion (Section~\ref{sec:generalized-directed-graph}).

The index generation method for DAG is further divided into two steps: 
1) Topological Compression and 2) Index Generation. 
We need to maintain information only during the first step, because the index generation 
step only builds index from the graphs generated in the topological compression (first) step. 
The topological compression step is described in Section~\ref{sec:tc}, where DAG is 
compressed iteratively by removing all odd-topology vertices and incident vertices. 
This compression process maintains loss of information using dummy edges, to keep the 
correct information we explicitly handle distance information for dummy edges as explained 
in Section~\ref{sec:dummyDist}. 

Lastly, as the structure of a converted DAG is different, we need to handle the queries 
a little differently. We have explained the modification of query processing in 
Section~\ref{sec:modified_query}. Hence, all the required logical modifications are handled 
and \alg\ maintains correctness for general directed graph.

\section{Experimental evaluation}
\label{sec:experimental-evaluation}

We compare performance of \alg\ with two of the recent methods 
(IS-Label and TreeMap) for answering distance query. 
We also compare \alg\ with baseline method Bidirectional Dijkstra's algorithm, 
which is one of the fastest online methods for 
single source shortest distance queries. 
For both IS-Label and TreeMap, codes are provided by their authors. 
For these experiments we use a machine with Intel 2.4 GHz processor, 
8 GB RAM and Ubuntu 14.04 LTS OS. 
In \cite{AED:Xiang:2014} the author has claimed that TreeMap works for weighted directed graphs, 
however we are provided with the code of unweighted version for TreeMap, 
hence all comparisons with TreeMap are for unweighted graphs. 
Additionally as Y. Xiang mentioned in the paper, 
TreeMap needs huge memory if tree width is above threshold (1000).
The only dataset we are able to run using above machine is WikiVote.
Hence for comparison with TreeMap, we used machine with AMD 2.3 GHz processor, 
132 GB RAM and Red Hat Enterprise Server Release 6.6 OS. 
We also perform comparison to IS-Lable using same machine 
for two datasets (Email\_Eu and Epinion). 
Using synthetic graphs of different sizes and degrees, 
we show that TreeMap is not scalable for higher degree graphs. 
To generate these synthetic graphs we use python package 
networkx (procedure name, \proc{Fast\_gnp\_random\_graph()}).

\subsection{Datasets}

Here for our experiments, we used seven real world datasets 
(Table~\ref{table:dataset}) from different domains to show wide applicability
of \alg\ .  $\vert V \vert$ and $\vert E \vert$ are the number of vertices and
the number of edges respectively. Similarly $\vert V_{DAG} \vert$ and $\vert
E_{DAG} \vert$ are the number of vertices and the edges in the DAG of the
corresponding graph. \textit{AD} and $AD_{DAG}$ are average degree values for
the graph and its DAG counterpart, respectively.  \textit{MD} and $MD_{DAG}$
are maximum degrees i.e. maximum in or out degree in the graph and its DAG,
respectively.  \textit{Largest SCC} is a size of the biggest DAG node which
encapsulate the maximum number of input graph nodes.


We collected all datasets from SNAP (Stanford Network Analysis Project) web page\footnotemark[5] 
except $Epinion$ trust network dataset, 
which we collected from \cite{TAB:Massa:2006}.
$AS\_Caida$ is a business relationship network and 
$Email\_Eu$ is a snapshot of an email network generated by European Research Institute. 
$Epinion$ dataset is a trust network generated from social network users, 
it represents which user trusts whom. 
$Gnutella$ is a peer-to-peer file sharing network where $Gnutella09$ is 
a snapshot of the network on 9th August 2002 and $Gnutella31$ is 
a snapshot of the same network on 31st August 2002. 
$WikiVote$ is a network generated from Wikipedia admin voting history data.

\footnotetext[5]{http://snap.stanford.edu/data/index.html}

\subsection{Results and Discussion}
\label{sec:result-discussion}

\begin{figure*}[ht]
\centering
    \begin{subfigure}[ht]{0.5\textwidth}
	\centering
	\includegraphics[width=75mm,keepaspectratio]{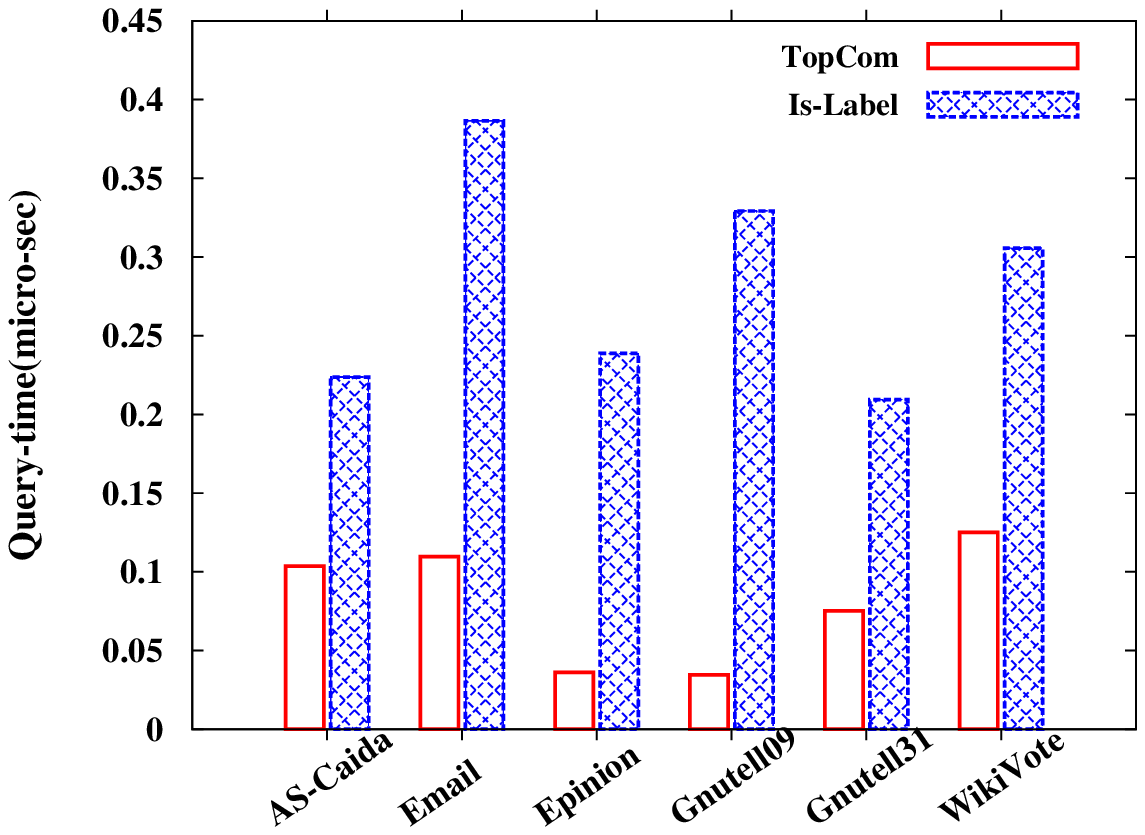}

	\caption{For DAG ($\mu$s)}	
	\label{fig:DAG-results}
    \end{subfigure}%
    ~
    \begin{subfigure}[ht]{0.5\textwidth}
	\centering

	\includegraphics[width=75mm,keepaspectratio]{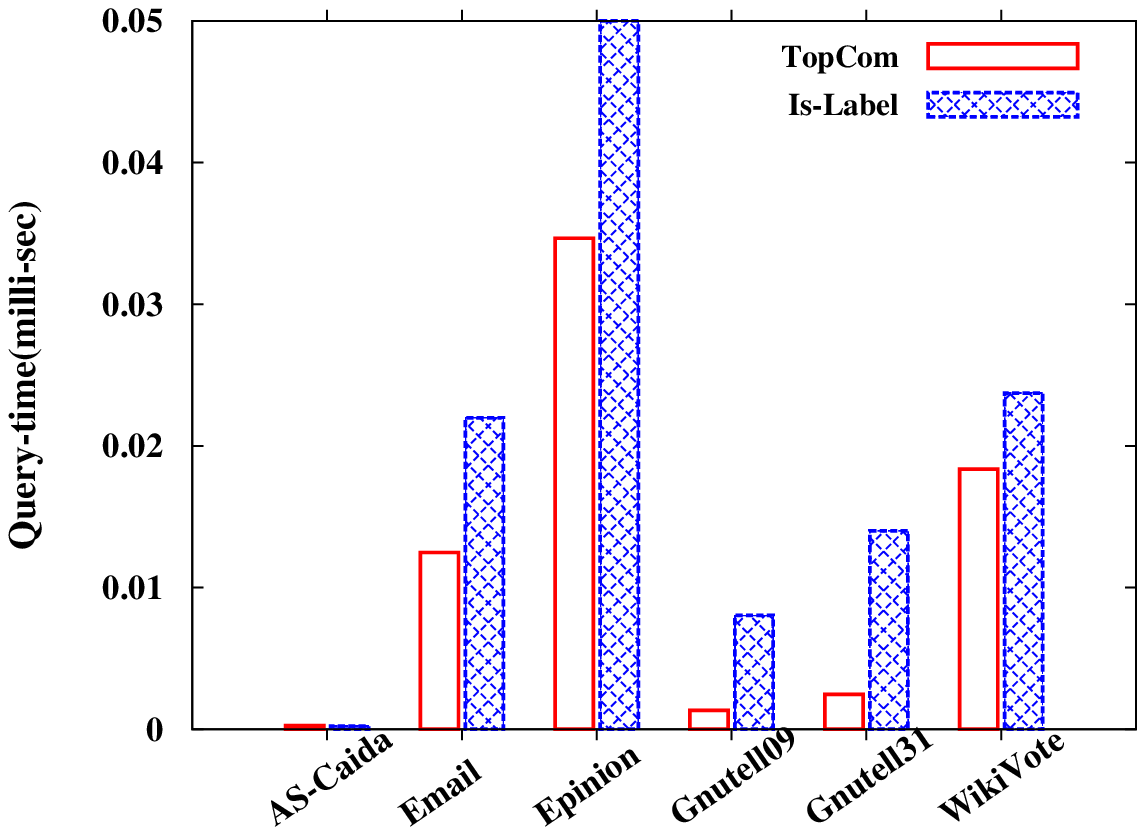}
	
	\caption{For General Graph (ms)}	
	\label{fig:GenGraph-results}
    \end{subfigure}
    	\vspace{-0.1in}
    \caption{Average Query time comparison}
    \hfill
\end{figure*}
\begin{figure*}[t!]
\centering
    \begin{subfigure}[t]{0.5\textwidth}
        \centering
		\includegraphics[width=75mm,keepaspectratio]{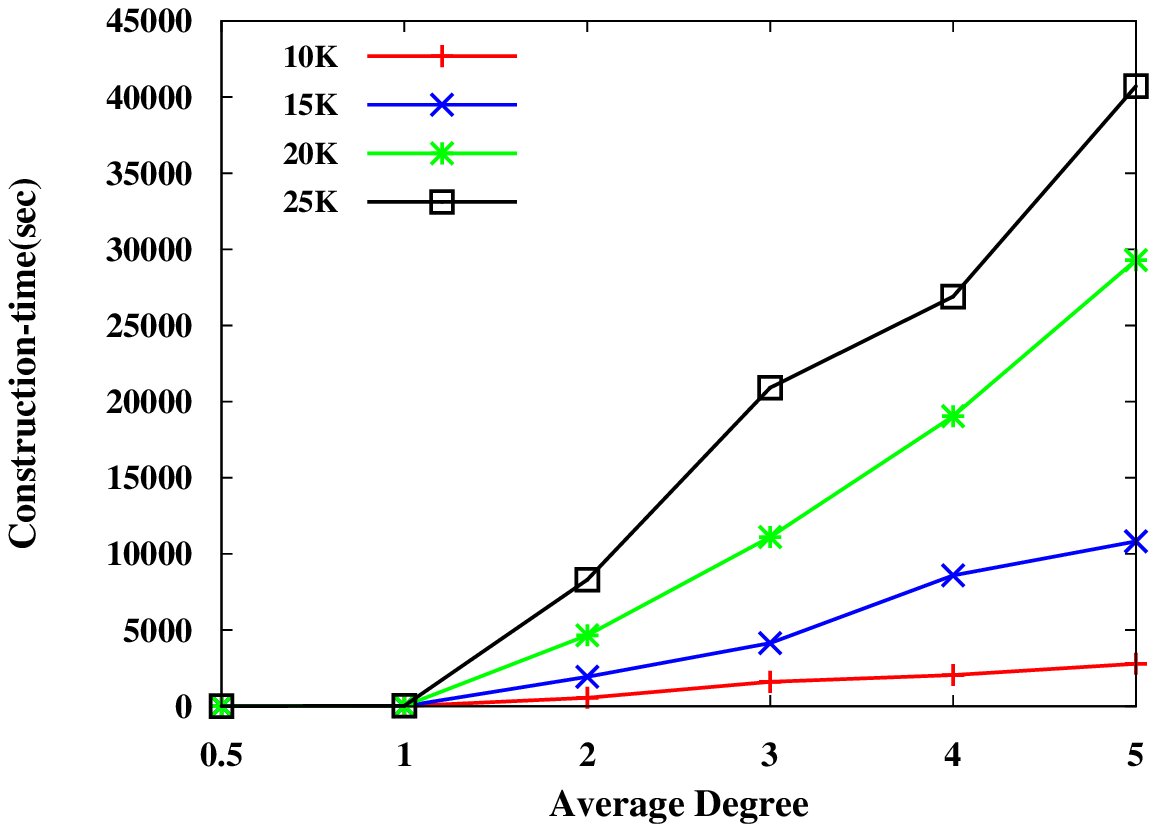}
        \caption{TreeMap Results on Synthetic Graphs}
        \label{fig:TreeMap building time}
    \end{subfigure}%
    ~
    \begin{subfigure}[t]{0.5\textwidth}
        \centering
		\includegraphics[width=75mm,keepaspectratio]{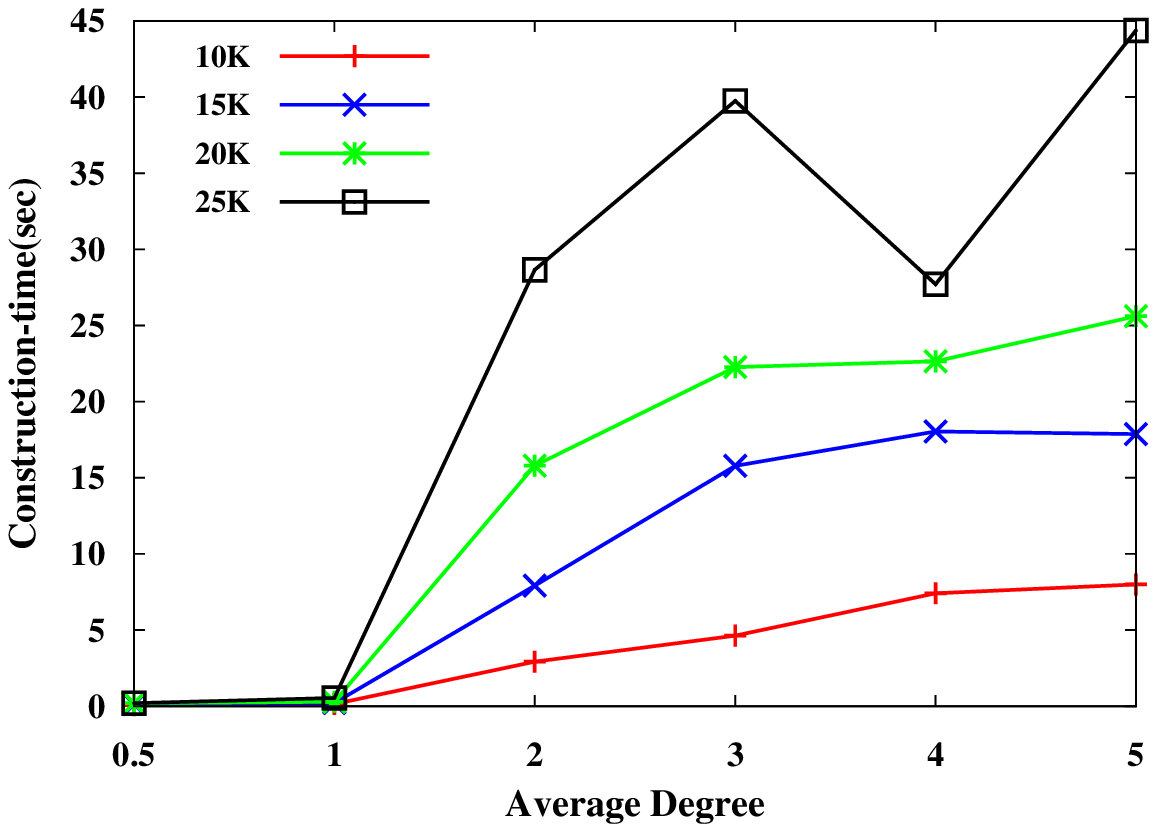}
        \caption{TopCom Results on Synthetic Graphs}
        \label{TopCom building time}
    \end{subfigure}
    \hfill
        	\vspace{-0.1in}
\caption{Index Building time for Synthetic graphs}
\end{figure*}

AS per expectation for DAG \alg\ outperforms IS-Label method for 
all datasets depicted in figure~\ref{fig:DAG-results}. 
Results are average query time over 10 times execution in micro-second ($\mu$s), 
where each method calculated 10K random queries in every execution. 
For more detailed comparison, if we look at table~\ref{table:results}, 
we can see that for datasets $Epinion, Gnutella09$ and $Gnutella31$ \alg\ 
outperforms IS-Label and TreeMap by an order of magnitude. 
For other datasets also \alg\ performs 2-3 times better than both of the competing methods. 
If we look at the Bi-Dijkstra results, \alg\ performs multiple orders of 
magnitude better for all datasets.
In figure~\ref{fig:GenGraph-results} results for general graphs are plotted, 
which clearly demonstrate superiority of \alg\ over IS-Label for general graph.
Here also we used Average Query time over 10 times execution
with each run calculating 10K random queries. 
Now if we look at table~\ref{table:IS-cmp} for dataset 
$AS\_Caida$ IS-Label performs better than \alg\, 
however the difference is $0.00003989$ $ms$ which is really very small. 
On the other hand \alg\ outperforms IS-Label for all other datasets. 
For $Gnutella31$ dataset \alg\ performs an order of magnitude better than IS-Label and 
surprisingly IS-Label performs really poor on $Epinion$ dataset such that \alg\ 
outperforms IS-Label by two orders of magnitude. 
For $Guntella09$ \alg\ performs almost 7 times better and 
for remaining datasets \alg\ performs almost two times better than IS-Label. 
Here once again \alg\ is multiple orders of magnitude faster than Bi-Dijkstra for all datasets.

Table~\ref{table:IS-cmp} shows result of TreeMap and \alg\ 
comparisons on unweighed graph. 
It is clear that \alg\ is competitively better in $AS\_cadia$, 
$Gnutella09$ and $Gnutella31$ datasets, but TreeMap performs an order of magnitude better 
for other three datasets. 
However, when we run the TreeMap for building index it took hours to 
build the index for some datasets, for example average index building time 
for $Epinion$ dataset is more than 9 hours, while $Gnutella31$ takes more than 26 hours. 
We believe one of the reasons is a bigger graph with higher average degree.
To find out the actual cause, we generated synthetic graphs of different sizes (10000-25000)
and degrees (0.5-5) and tried to build indexes using TreeMap. 
In figure~\ref{fig:TreeMap building time} the index building time 
is shown in seconds, after degree 1 all graphs started taking higher time 
and for bigger graphs the slope of the curve is very large. 
We compare construction time of \alg\ for the same set of synthetic graphs shown 
in figure~\ref{TopCom building time}, and as shown \alg\ hardly takes few 
seconds for index construction. 
Highest time taken by \alg\ is 44 sec for 25K nodes graph with average degree 5, 
which is almost thousand times faster compared to TreeMap. 
From this we can easily conclude that it is difficult for TreeMap to scale for higher degree graphs.

\section{Conclusions}
In this paper we proposed \alg\ : a unique indexing method to answer distance query 
for directed real-world graphs. 
This method uses topological ordering property of DAG and 
describes a novel method for distance preserving compression of DAG. 
We compared \alg\ with IS-Label and found the our method performs better than 
IS-Label for both weighted DAG and weighted general graph. 
We strongly believe our method should perform similar or better for unweighed graphs, 
because we store distance information in label irrespective of weighted/unweighted edges. 
We  do not compare \alg\ with other recent methods such as HCL \cite{AHC:Jin:2012} 
and state-of-art 2-Hop \cite{2Hop:Cohen:2002}, 
because Fu et al. have compared IS-Label with HCL and 
proved superiority of IS-Label in \cite{ISL:Fu:2013}. 
We also compare \alg\ with the recent TreeMap method, which performs better for some datasets,
however, we show that this method is not scalable for huge graphs with higher degree.
We plan to study further to build index for dynamic large graphs 
that can answer exact distance query in an acceptable time.

\bibliographystyle{ACM-Reference-Format-Journals}
\bibliography{acmsmall-sample-bibfile}

\elecappendix

\medskip

This paper is an extension of the short paper published in an international conference~\cite{TIS:dave:2015}.
In this journal version paper, we have made considerable additions to the short paper, 
which are listed below:

\begin{enumerate}
\item We have added an important theoretical explanation with proof that our \alg\ is the 2-Hop indexing method, which is the most accepted indexing method from the last decade. We used the well known mathematical induction technique for this proof which is explained in the section~\ref{sec:proof}.

\item Our short paper version included indexing for only DAG (Directed Acyclic Graph). In the journal version we have extended our indexing method for arbitrary directed graphs. The description of the method is available in section~\ref{sec:generalized-directed-graph}.

\item In section~\ref{sec:experimental-evaluation} we have added experimental evaluation of the extended method for any directed graph that shows considerable improvement over state of the art methods.

\item We have also added experimental study to compare scalability of TopCom with TreeMap. The comparison is performed over graphs with $10000, 15000, 20000$ and $25000$ nodes and their average degrees being $0.5, 1 ,2, 3, 4$ and $5$. So $24$ graphs with different sizes and densities were used for checking the scalability of the indexing method. This comparative study shows that TopCom is significantly better than TreeMap for huge graphs. The description is available in section~\ref{sec:result-discussion}.

\item A newly added section~\ref{sec:related-works} is a related work section which has brief descriptions of the recent research on indexing. 

\end{enumerate}

We believe there is atleast 45-50\% additional content in the journal version as compared to the short paper.

\end{document}